\newif\ifarxiv
\arxivtrue          

\ifarxiv
  \documentclass[aps,prd,reprint,amsmath,amssymb,superscriptaddress,nofootinbib,floatfix]{revtex4-2}
\else
  \documentclass[pdflatex,sn-mathphys-num]{sn-jnl}
\fi

\usepackage{amsmath,amssymb,amsfonts,mathtools}
\usepackage{amsthm}
\usepackage{bm}
\usepackage{booktabs}
\usepackage{array}
\usepackage{microtype}
\usepackage{dsfont}
\usepackage{tikz}
\usepackage{enumitem}
\usepackage{xcolor}
\usepackage[colorlinks=true,allcolors=blue]{hyperref}

\theoremstyle{plain}
\newtheorem{theorem}{Theorem}[section]
\newtheorem{proposition}[theorem]{Proposition}
\newtheorem{corollary}[theorem]{Corollary}

\newtheorem{assumption}[theorem]{Assumption}

\theoremstyle{remark}
\newtheorem{remark}[theorem]{Remark}

\theoremstyle{definition}

\newcommand{\cH}{\mathcal{H}}

\newcommand{\dd}{\mathrm d}

\begin{document}

\newcommand{\PaperTitle}{The Emergence of Time from Quantum Records: Actualization and the Surprisal Clock}
\newcommand{\PaperKeywords}{foundations of quantum mechanics, emergent time, quantum records, surprisal, relational clocks}
\newcommand{\PaperAbstract}{%
We develop a record-based account of internal time in quantum mechanics. Persistent records first define an ordinal chronology through inclusion of their accumulated Boolean algebras. On a specified record filtration, we prove that, under three minimal consistency requirements, namely dependence only on conditional Born weight, invariance under sequential refinement of the same recorded fact, and continuity, the actualization of each outcome contributes an internal time proportional to its \textit{surprisal}, defined as the negative logarithm of that probability. A certain outcome therefore contributes no duration, whereas less probable outcomes contribute larger increments. The mean and variance of the accumulated clock are governed by the Shannon entropy and varentropy of the record process, its moment-generating function is related to the R'enyi entropy spectrum, and its pathwise fluctuations admit a Doob decomposition into a predictable entropic compensator and a martingale. Two further results provide independent consistency checks. Within the stated class of finite-dimensional bipartite states, universal additivity of local clock readings across all admissible local record contexts is equivalent to the absence of entanglement. In the quantum Zeno regime, the number of monitoring rounds may diverge while the expected accumulated surprisal tends to zero. The construction therefore supplies a canonical internal-time functional for a specified quantum record process without presupposing an external clock. It does not by itself identify which physical record processes constitute material clocks; determining their production rates and dynamical realization remains a separate problem.%
}

\ifarxiv
  \title{\PaperTitle}
  \author{Lionel Martellini}
  \email{lionel.martellini@edhec.edu}
  \affiliation{EDHEC Quantum Institute, EDHEC Business School, Nice, France}
  \begin{abstract}
  \PaperAbstract
  \end{abstract}
\else
  \title[The Emergence of Time from Quantum Records]{\PaperTitle}
  \author*[1]{\fnm{Lionel} \sur{Martellini}}\email{lionel.martellini@edhec.edu}
  \affil[1]{\orgdiv{EDHEC Quantum Institute},
  \orgname{EDHEC Business School},
  \orgaddress{\city{Nice}, \country{France}}}
  \abstract{\PaperAbstract}
\fi

\maketitle

\section{Introduction}
\label{sec:introduction}

Time plays several logically distinct roles in quantum theory. In the standard
formalism and in laboratory applications, it usually appears as an
\textit{external} parameter labeling state evolution. This external role should be distinguished from the stronger Newtonian notion of a universal
absolute time: relativistic physics does not permit such a preferred temporal
structure, although frame-dependent coordinate times may of course still be
used as external parameters. A different possibility is that time is not fundamental in this sense, but
that its operational structure emerges from physical processes
\textit{internal} to the quantum description itself.  In generally covariant
and quantum-gravitational settings, this emergent construction is a
necessity since a closed system may satisfy a stationary constraint with no
preferred external time
\cite{DeWitt1967,Kuchar1992,Isham1993}. A central challenge is therefore to
explain how temporal structure can be grounded in physical processes internal
to a quantum system, rather than simply presupposed through an external time
parameter.

The present paper addresses a narrower question: once physical events have produced a persistent sequence of records, can temporal order and quantitative duration be reconstructed from the record process itself? If so, an operational temporal structure would emerge from the accumulation of physical facts rather than being supplied as an external parameter. Relational-clock constructions address one part of this problem. The
Page--Wootters mechanism and related constructions describe effective evolution
relative to an internal subsystem, including globally stationary settings,
while thermal time, partial observables, and quantum-reference-frame
approaches provide different ways of replacing an externally privileged time
parameter by relational structure
\cite{PageWootters1983,Wootters1984,ConnesRovelli1994,Rovelli1991,
Giovannetti2015}.  These approaches establish that an external clock need not
be fundamental, but they do not in general supply a unique rule for assigning
a numerical duration to the successive factual records generated by a given
physical clock process.  A complementary body of work concerns histories and records.  Consistent- and
decoherent-history formulations assign probabilities to suitably decoherent
sequences of alternatives, while environmental decoherence and redundant
records explain how information about such alternatives can become stably
available
\cite{Griffiths1984,Omnes1992,GellMannHartle1993,Zurek2003}. These
frameworks provide much of the formal structure needed to discuss persistent
record histories, but they are not primarily theories of the numerical
duration carried by an individual realized history. 

The present paper connects these two strands of the literature by asking whether an operational notion of time can emerge from the realized record process itself: once a physical process produces a persistent sequence of records, is there a canonical way of assigning a quantitative duration to that history? This question is naturally motivated by the theory of potentiality introduced in Ref.~\cite{MartelliniPotentiality}, where the quantum state represents unresolved physical potentialities and measurement corresponds to their actualization into stable factual records. The temporal hypothesis explored here is that the accumulation of such actualized facts can support not only an ordering of events but also a quantitative clock structure. The question is then what additive clock functional, if any, is compatible with their Born probabilities and with refinement of the record. The formal analysis in this paper, however, does not require this ontological reading. We assume a realized persistent record history and do not address the dynamical mechanism by which one possible outcome becomes factual. The representation theorem and the subsequent analysis use only standard quantum states, Born probabilities, conditional updates, and persistent records. The results can therefore be assessed independently within the standard formalism of quantum mechanics.

Within this scope, our main findings can be stated as follows. Persistence of
records first provides an ordinal chronology, since later record algebras retain
the factual content fixed by earlier ones. Once a record filtration has been
specified, the representation theorem then identifies a quantitative pathwise
clock. If the increment associated with a newly realized record depends only on
its conditional Born probability, is additive under sequential refinement of
the same final recorded fact, and satisfies a mild continuity condition, its
dependence on probability is shown to be logarithmic. Equivalently, each
actualized outcome contributes a duration proportional to its \textit{surprisal},
defined as the negative logarithm of its conditional probability. An outcome
with conditional probability one therefore contributes no increment, whereas
increasingly unlikely outcomes contribute progressively larger ones. For a realized history $h$ with Born probability $P(h)$, the accumulated
clock reading takes the compact form $-\sigma\ln P(h)$, where $\sigma>0$
is a free calibration parameter, to be fixed for a given physical clock
process by matching its mean reading with respect to a chosen reference time scale. Shannon entropy,
varentropy, R\'enyi quantities, and martingale and large-deviation results are then introduced to provide a detailed statistical characterization of the surprisal clock distributions across possible paths. For projective records, the surprisal increment also admits an exact geometric
representation in terms of the Fubini--Study angle between the pre- and
post-record rays. This analysis shows that a nonzero surprisal-clock
reading must come from the formation of physical records, rather
than from an arbitrary subdivision of continuous state evolution. While the construction is universal in form, its realization is relative to a specified record filtration and the clock functional 
itself does not single out any admissible filtration as privileged. Filtrations that are merely different
descriptions of the same factual record process give the same dimensionless
surprisal reading, whereas genuinely different record processes generally define
different clocks. When two filtrations are jointly representable, their readings
can nevertheless be related through the conditional and joint Born probabilities
of their records. In this sense, the construction admits a
filtration-covariant comparison structure, developed in
Sec.~\ref{subsec:filtration-covariance}. A separate issue concerns physical realization. The representation theorem determines the
clock functional once the record process has been specified, but not the dynamics
that generates that process or its physical rate. Whether a given system produces
persistent, stable, and sufficiently regular records to function as a material
clock must therefore be determined from its underlying dynamics. 

Two further results clarify how the construction behaves in situations involving correlations and repeated measurements. First, for pure bipartite states, we show that additivity of the two local clock readings in every pair of local measurement contexts is equivalent to the absence of entanglement. Thus, universal additivity of the two local clock readings holds if and only if the state is a product state. For entangled states, the two local readings remain individually well defined, but combining them requires the joint record process, with the departure from additivity determined by the correlations between the two records. Second, in the quantum Zeno regime, the number of measurement rounds can diverge while the total accumulated surprisal tends to zero. As successive outcomes become almost certain, observation continues to generate records, but those records contribute essentially no additional surprisal to the clock. In this sense, the Zeno limit suppresses observable transitions while the corresponding surprisal clock also freezes, even though a finite accumulated forward projective motion remains. 

In addition to the two strands of the literature discussed above, namely relational
approaches to time on the one hand, and consistent- or decoherent-history
formulations on the other, the present construction is also related to several existing approaches to internal time in quantum systems. First, quantum-reference-frame formulations emphasize that conditioned
states and dynamical generators may depend on the choice of internal clock
\cite{VanrietveldeEtAl2020,HoehnSmithLock2021}. The key difference is that  the present paper addresses
instead the record-based reading associated with a specified clock process
and does not attempt a general theory of transformations between nonideal
quantum clocks. On a different note, entropic dynamics also uses temporal succession to organize
inferential change \cite{Caticha2011}, while other entropy-based clock
proposals employ coarse-grained entropies as internal parameters
\cite{Weberszpil2026,Barontini2026}. The distinction here is that the
primitive quantity is the pathwise surprisal of each realized record;
Shannon entropy enters only after averaging over alternative histories. Finally, our construction should also be distinguished from approaches based on time
observables or thermodynamic entropy production. The surprisal clock is a
classical random variable defined on a record filtration, not a self-adjoint
operator canonically conjugate to the system Hamiltonian, so Pauli-type
obstructions are not directly engaged
\cite{Pauli1933,BuschGrabowskiLahti1994}. Likewise, the Shannon entropy of
record production is not identified with thermodynamic entropy production:
relating record formation to heat, erasure cost, or irreversible
thermodynamics requires an additional physical model of the recorder and its
environment. 

The rest of the paper is organized as follows. Section~\ref{sec:records} introduces
records and realized histories, and Sec.~\ref{sec:ordinal} derives their
ordinal structure. Section~\ref{sec:actualization-time} proves the
representation theorem and develops the main statistical properties of the
resulting surprisal clock. Section~\ref{sec:moments-deviations} develops its
large-deviation structure and relation to the R\'enyi entropy spectrum.
Section~\ref{sec:projective-time} relates record surprisal to projective
Hilbert-space geometry and examines its continuum scaling.
Section~\ref{sec:multiple} studies the comparison and additivity of clocks
associated with different record processes, while Sec.~\ref{sec:zeno-time}
analyzes the effect of repeated and continuous monitoring on the clock rate. Sec.~\ref{sec:further-research} discusses ontological implications and open directions, and
Sec.~\ref{sec:conclusion} summarizes the scope of the result.

\section{Quantum actualization, records, and histories}
\label{sec:records}

For simplicity, we work throughout with discrete record processes, with finite
or countable outcome sets. Probabilities are the standard Born probabilities
associated with the relevant quantum measurement or record process. 

\subsection{Actualization and conditional state update}
\label{subsec:conditional-potentiality}

Following the terminology of the theory of potentiality
\cite{MartelliniPotentiality}, throughout this paper we use
\emph{actualization} as an operational shorthand for the formation of a stable
record associated with a realized outcome. Here a record is called stable if,
once formed, its value remains available as fixed conditioning data for the
subsequent stages of the record process under consideration. This notion of stability is deliberately operational: a record need only remain
sufficiently persistent and distinguishable over the timescale and set of
operations relevant to the process. In realistic systems this stability is
generally approximate, with imperfect distinguishability and residual coherence
\cite{Zurek2003,Schlosshauer2005,OllivierPoulinZurek2004}. Environmental
decoherence provides a standard mechanism for suppressing coherent
recombination and supporting an approximately classical record structure.
Accordingly, not every state-changing interaction constitutes an actualization;
the representation theorem below assumes a specified idealized record
filtration and does not require a microscopic model of record formation.

Suppose that a measurement or interaction produces a record with possible values
$r\in\Omega_R$. The recorded alternatives are represented by a family of mutually
orthogonal projectors $\{\Pi_r\}_{r\in\Omega_R}$ on the system Hilbert space, one for
each value of the record, not necessarily of rank one, with
\begin{equation}
\sum_{r\in\Omega_R}\Pi_r=\mathds{1}.
\end{equation}
For a projective measurement, conditional on the formation of a record with
value $r$, the state used for subsequent predictions is
\begin{equation}
|\psi_r\rangle
=
\frac{\Pi_r|\psi\rangle}{\sqrt{p(r)}},
\qquad
p(r)=\langle\psi|\Pi_r|\psi\rangle,
\label{eq:conditioning}
\end{equation}
whenever $p(r)>0$. Equation~\eqref{eq:conditioning} is a conditional
state-update rule: it specifies the predictive state \emph{given} the newly
recorded fact $r$.

The same rule applies recursively along a realized record sequence.
Let $r_n$ denote the value of the stable record formed at stage $n$, and write
$|\psi_{h_n}\rangle$ for the quantum state conditioned on the records accumulated
up to that stage. Then
\begin{equation}
|\psi_{h_n}\rangle
=
\frac{
\Pi_{r_n}|\psi_{h_{n-1}}\rangle
}{
\sqrt{p(r_n\mid h_{n-1})}
},
\qquad
p(r_n\mid h_{n-1})
=
\langle\psi_{h_{n-1}}|\Pi_{r_n}|\psi_{h_{n-1}}\rangle .
\label{eq:history-conditioning}
\end{equation}
The history notation $h_n$ is defined explicitly in the next subsection. More
general quantum instruments replace the projective numerator in
Eq.~\eqref{eq:history-conditioning}, while retaining the distinction between
the accumulated factual records and the predictive quantum state conditioned
on those records.

Thus actualization has two logically distinct consequences within the standard
formalism used here. It adds new factual content to the persistent record
structure, while changing the conditional quantum state used to assign Born
probabilities to subsequent records. The accumulated records constitute the
growing factual history, whereas the conditional quantum state determines the
Born weights assigned to subsequent records along that history. The update
is a conditioning on realized records and does not, by itself, introduce an
independent physical time parameter. The possible ontological reading of this
conditional sequence in terms of an underlying atemporal domain of
potentialities is deferred to Sec.~\ref{sec:further-research}.

\subsection{Histories and accumulated record algebras}

Consider a sequence of actualized records
\begin{equation}
h_n=(r_1,r_2,\ldots,r_n),
\qquad
r_k\in\Omega_{R_k},
\label{eq:history}
\end{equation}
where $\Omega_{R_k}$ denotes the set of possible values of the $k$th record.
The Born likelihood of this history is
\begin{equation}
P(h_n)
=
\prod_{k=1}^{n} p(r_k\mid h_{k-1}),
\qquad h_0=\varnothing,
\label{eq:history-probability}
\end{equation}
where each conditional probability is computed from the state conditioned on the previous
records. Equation~\eqref{eq:history-probability} is simply the chain rule for the sequential
measurement process.

Let $\Sigma_n$ denote the Boolean algebra generated by the first $n$ records. The physical
assumption required for an internal temporal order is the persistence of these records.

\begin{assumption}[Record persistence]
\label{ass:persistence}
The stability condition introduced above is assumed to
hold jointly along a realized history: once a record has been formed, its
factual content remains part of the accumulated record structure at all later
stages. Accordingly, the record algebras are nested,
\begin{equation}
\Sigma_0\subseteq\Sigma_1\subseteq\cdots\subseteq\Sigma_n\subseteq\cdots.
\label{eq:nested-records}
\end{equation}
Later actualizations may refine or enlarge the recorded information, but they
do not erase the factual content fixed by earlier records.
\end{assumption}

Assumption~\ref{ass:persistence} is not implied by microscopic unitary
dynamics alone. Rather, it describes the effective regime in which interactions
with additional degrees of freedom produce records whose factual content is
sufficiently stable over the history under consideration.

\section{Ordinal time from record inclusion}
\label{sec:ordinal}

The first temporal structure supplied by persistent records is an order.

\begin{proposition}[Record order]
\label{prop:record-order}
Under Assumption~\ref{ass:persistence}, the accumulated record algebras along one realized history are linearly ordered by inclusion,
\begin{equation}
 \Sigma_m\preceq\Sigma_n
 \quad\Longleftrightarrow\quad
 \Sigma_m\subseteq\Sigma_n.
 \label{eq:record-order}
\end{equation}
For a countable history this order admits an order-preserving embedding into the real line.
\end{proposition}

\begin{proof}
Set inclusion is reflexive, antisymmetric, and transitive, and record persistence makes any two algebras on the same history comparable. Every countable linear order embeds into the rationals and hence into the reals.
\end{proof}

The proposition defines temporal order, that is a before/after relation, without yet assigning quantitative duration intervals. The realized records also carry compatible Boolean valuations. Let $v_n:\Sigma_n\to\{0,1\}$ assign truth values to propositions certified by the first $n$ records.

\begin{theorem}[Growing record structure]
\label{thm:growing-records}
Under Assumption~\ref{ass:persistence},
\begin{equation}
 \Sigma_n\subseteq\Sigma_{n+1},
 \qquad
 v_{n+1}|_{\Sigma_n}=v_n.
 \label{eq:growing-record-system}
\end{equation}
The inclusion is strict exactly when the new record is not already determined by $\Sigma_n$.
\end{theorem}

\begin{proof}
The algebra inclusion is Assumption~\ref{ass:persistence}. Persistence of the earlier recorded values makes $v_{n+1}$ and $v_n$ agree on the generators of $\Sigma_n$, hence on the Boolean algebra they generate. A record already determined by $\Sigma_n$ adds no new proposition; a record not determined by it enlarges the algebra.
\end{proof}

Record persistence therefore describes a growing structure of determinate
facts: as the history develops, new factual content is added while earlier
recorded content is retained. This growth should not, however, be interpreted
as the progressive revelation of a single pre-existing assignment of values
to all quantum observables. In Hilbert-space dimension at least three, the
Kochen--Specker theorem excludes such a global noncontextual Boolean valuation
\cite{KochenSpecker1967}. The temporal order obtained above is therefore an
order of accumulated factual records, rather than merely an ordering of labels
or the uncovering of an already determinate history.

\begin{remark}[Factual record content versus predictive state]
\label{rem:factual-versus-predictive}
The valued Boolean algebra $(\Sigma_n,v_n)$ represents the determinate factual
content accumulated by stage $n$; it is not by itself a complete predictive
state. Within the standard quantum formalism used in this paper, subsequent
Born probabilities also depend on the quantum state conditioned on that factual
history. It is therefore useful to distinguish schematically
\begin{equation}
z_n
=
\bigl(\Sigma_n,v_n;|\psi_{h_n}\rangle\bigr),
\label{eq:record-conditioned-state}
\end{equation}
where $(\Sigma_n,v_n)$ carries what has become factual and
$|\psi_{h_n}\rangle$ is the quantum state conditioned on those facts.

Under a further actualization these two components change in complementary
ways:
\begin{equation}
(\Sigma_n,v_n)
\longrightarrow
(\Sigma_{n+1},v_{n+1}),
\qquad
|\psi_{h_n}\rangle
\longrightarrow
|\psi_{h_{n+1}}\rangle .
\label{eq:dual-actualization-update}
\end{equation}
The first update enlarges the persistent factual record; the second is the
standard conditional quantum-state update used to assign probabilities to
future records.

Two descriptions may therefore carry the same accumulated factual records
while differing in their conditioned quantum states and hence in the
statistics of future records. They are operationally indistinguishable only if
they agree on all admissible future record statistics. The companion
space--matter paper makes this operational equivalence explicit. Nothing in
the clock representation theorem requires $(\Sigma_n,v_n)$ alone to determine
those statistics; it requires only the conditional Born weights along the
specified record filtration.
\end{remark}

An order alone does not determine duration. We now ask whether a specified record history admits a canonical additive interval functional.

\section{Actualization time from sequential records}
\label{sec:actualization-time}

We now turn from ordinal temporal order to the construction of a quantitative pathwise duration.

\subsection{Requirements on a pathwise actualization-time increment}

Consider one actualization step. Immediately before the step, the conditioned
state assigns probability $p\in(0,1]$ to the outcome that is subsequently
recorded. Let $d(p)$ denote the clock increment assigned to this realized
actualization. We impose three requirements.

\begin{itemize}
\item[(R1)] \emph{Born-weight sufficiency.} Once the record filtration and
the physical calibration of the clock have been specified, the increment assigned to one realized record depends only on its conditional Born probability $p$.

\item[(R2)] \emph{Record-refinement invariance.} For any $p,q\in(0,1]$, compare a fine-grained recorded fact of Born weight $pq$ with a description in which the same final fact is reached through a coarse record of weight $p$ followed by a refinement of conditional weight $q$, the coarse proposition being implied by the final record. The clock class is required to be invariant under this refinement, so
\begin{equation}
 d(pq)=d(p)+d(q).
 \label{eq:additivity-functional}
\end{equation}

\item[(R3)] \emph{Regularity and nontriviality.}
The function $d$ is continuous and nonnegative, and is not identically zero.
\end{itemize}

Requirement (R1) is the main restriction defining the class of clocks studied
here. A more general physical clock could assign increments that depend not only
on the conditional Born weight but also on dynamical or implementation-specific
quantities, such as the interaction Hamiltonian, the duration of the recording
interaction, the stability of the memory, or its energetic cost. Such clocks
are not excluded as physical possibilities; they simply fall outside the
representation theorem proved below. The purpose of (R1) is to isolate the
part of the clock reading that can be determined from the realized record
statistics alone. This does not imply that dynamics is irrelevant to physical timekeeping.
The underlying dynamics determines whether records are produced, how stable
they are, and at what physical rate they occur. Requirement (R1) states only
that, once a particular record process and its calibration have been specified,
the increment associated with an individual realized record carries no
additional event-dependent information beyond its conditional Born weight.
Fixed properties of the physical clock may still enter through the overall
calibration scale $\sigma$. Requirement (R2) is an additivity condition stating that when two descriptions encode the same nested factual refinement, the reading may not depend on where the description inserts the intermediate record. The theorem is therefore a representation theorem for a specific class of clocks on a \emph{given} record filtration, not a derivation of every notion of time and not a rule selecting the physical clock mechanism.

    \begin{theorem}[Representation theorem for a surprisal record clock]
    \label{thm:surprisal}
    The unique nontrivial function satisfying (R1)--(R3), up to an overall
    positive multiplicative constant, is
    \begin{equation}
    d(p)=-\sigma\ln p,
    \qquad \sigma>0.
    \label{eq:surprisal}
    \end{equation}
    The constant $\sigma$ has units of time and fixes the overall calibration of
    the clock. Consequently, the accumulated actualization-time reading assigned to a record history $h_n$, denoted by $\tau_{\rm A}(h_n)$, is
    \begin{equation}
    \tau_{\rm A}(h_n)
    =
    \sum_{k=1}^{n}d\bigl(p(r_k\mid h_{k-1})\bigr)
    =
    -\sigma\ln P(h_n),
    \label{eq:pathwise-clock}
    \end{equation}
    where $P(h_n)$ is the Born probability of the realized record history defined in Eq.~\eqref{eq:history-probability}.
    \end{theorem}

\begin{proof}
Define $f(x)=d(e^{-x})$ for $x\geq0$. Requirement (R2) gives
$f(x+y)=f(x)+f(y)$. Continuity implies $f(x)=cx$ for a constant $c$. Nonnegativity gives
$c\geq0$, and the calibration in (R3) fixes $c=\sigma$. Equation~\eqref{eq:pathwise-clock}
then follows from the probability chain rule in Eq.~\eqref{eq:history-probability}.
\end{proof}

The function $-\ln p$ is the \emph{self-information} of Shannon
\cite{Shannon1948}, for which Tribus coined the term \emph{surprisal}
\cite{Tribus1961}: it vanishes for a certain event, diverges for an
arbitrarily rare one, and is the unique continuous assignment for which
independent surprises add. The present theorem gives the same functional a
different reading: not information delivered to an observer, but an additive actualization-time
increment registered by the record process. The logarithm is forced by the multiplicative composition of conditional probabilities. The
same mathematical structure underlies the additivity of self-information in information
theory. Here it is interpreted as an additive path functional on realized record histories: an outcome that was
certain contributes no increment, while an outcome of small prior probability contributes a
larger increment. 

The representation theorem fixes the dependence of the clock
increment on the conditional Born weight, but not its absolute time scale.
For a given physical record process, $\sigma$ must be fixed by calibration,
for example by matching the mean clock reading to a chosen reference time
scale. The clock reading defined by the theorem is also relative to the specified record
filtration. Enlarging that filtration by adding genuinely new recorded
information will in general change the reading; this should be distinguished
from a mere redescription of the same factual record process. The comparison
of clocks associated with different filtrations is developed in
Sec.~\ref{subsec:filtration-covariance}.

\subsection{Ensemble actualization time and Shannon entropy}

The pathwise actualization time in Eq.~\eqref{eq:pathwise-clock} fluctuates from one history to another. Write $\tau_n:=\tau_{\rm A}(R_1,\ldots,R_n)$ for the corresponding random variable on histories. Averaging it over all histories of depth $n$ gives an ensemble clock, which we analyze now. In what follows and throughout, uppercase $R_k$ denotes the record outcome at step $k$
regarded as a random variable under the Born measure on histories, with
realized values $r_k$; a history $h_n=(r_1,\ldots,r_n)$ is thus one
realization of $(R_1,\ldots,R_n)$, and $H(\cdot)$ and $H(\cdot\mid\cdot)$
are the Shannon entropy and conditional entropy of these variables.

\begin{proposition}[Shannon mean of the pathwise clock]
\label{prop:shannon-mean}
Let $\cH_n$ be the set of record histories of length $n$. Then
\begin{equation}
\mathbb{E}[\tau_n]
=
-\sigma\sum_{h_n\in\cH_n}P(h_n)\ln P(h_n)
=
\sigma H(R_1,\ldots,R_n),
\label{eq:mean-clock}
\end{equation}
where $H$ denotes Shannon entropy in natural units. Equivalently,
\begin{equation}
\mathbb{E}[\tau_n]
=
\sigma\sum_{k=1}^{n}H(R_k\mid R_1,\ldots,R_{k-1}).
\label{eq:conditional-entropy-sum}
\end{equation}
\end{proposition}

\begin{proof}
Substituting Eq.~\eqref{eq:pathwise-clock} into the expectation gives the first equality.
The second is the definition of Shannon entropy. Equation~\eqref{eq:conditional-entropy-sum}
follows from the Shannon chain rule.
\end{proof}

Thus the primary clock is pathwise surprisal, while Shannon entropy is its ensemble mean. This distinction avoids identifying the actualization-time reading along one realized record
history with an entropy of an ensemble that also contains unrealized alternatives.

\subsection{Fluctuations and R\'enyi entropies}

The mean clock reading captures only the first moment of the random variable
$\tau_n$. Its full probability distribution contains additional information
about the fluctuations of realized clock histories around that mean. A natural
way to characterize this distribution is through the moment-generating function
of the dimensionless clock variable
\[
S_n:=\frac{\tau_n}{\sigma}=-\ln P(H_n),
\]
where $H_n$ denotes the random realized history. Define
\begin{equation}
M_n(s):=\mathbb E[e^{sS_n}],
\qquad
K_n(s):=\ln M_n(s),
\end{equation}
whenever these quantities exist. The derivatives of $M_n$ at $s=0$ generate
the moments of $S_n$, while the derivatives of $K_n$ generate its cumulants:
\begin{equation}
M_n^{(m)}(0)=\mathbb E[S_n^m],
\qquad
K_n^{(m)}(0)=\kappa_m(S_n).
\end{equation}
Thus $K_n'(0)$ gives the mean, $K_n''(0)$ the variance, and higher derivatives
encode the higher-order fluctuations of the realized clock. For the surprisal clock, these generating functions are determined directly
by the R\'enyi entropy spectrum of the record distribution. Indeed, as shown in the proposition
below,
\begin{equation}
M_n(s)
=
\sum_{h_n}P(h_n)^{1-s},
\qquad
K_n(s)
=
sH_{1-s}(P).
\end{equation}
The Shannon entropy appears as the
first cumulant in the limit $s\to0$, the varentropy as the second cumulant,
and the higher derivatives of the R\'enyi transform determine the higher
cumulants of the clock. In this sense, the R\'enyi spectrum characterizes the full distribution of pathwise clock fluctuations.

\begin{proposition}[R\'enyi transform of the clock]
\label{prop:renyi}
For a finite history ensemble, or more generally whenever the following moment exists,
\begin{equation}
\left\langle e^{s\tau_n/\sigma}\right\rangle
=
\sum_{h_n\in\cH_n}P(h_n)^{1-s}
=
\exp\!\left[sH_{1-s}(P)\right],
\label{eq:renyi-transform}
\end{equation}
where $H_\alpha(P)$ is the R\'enyi entropy of order $\alpha$. In particular,
\begin{equation}
\left\langle e^{-\tau_n/\sigma}\right\rangle
=
\sum_{h_n}P(h_n)^2
=
e^{-H_2(P)}.
\label{eq:collision-identity}
\end{equation}
\end{proposition}

\begin{proof}
By Eq.~\eqref{eq:pathwise-clock},
$e^{s\tau_n/\sigma}=P(h_n)^{-s}$. Multiplication by the sampling weight $P(h_n)$ and summation
over histories gives $\sum_hP(h)^{1-s}$. The R\'enyi identity follows from
\begin{equation}
H_\alpha(P)=\frac{1}{1-\alpha}\ln\sum_hP(h)^\alpha
\end{equation}
with $\alpha=1-s$.
\end{proof}

Differentiating the logarithm of Eq.~\eqref{eq:renyi-transform} at $s=0$
recovers the mean in Eq.~\eqref{eq:mean-clock}. The second cumulant is
\begin{equation}
\operatorname{Var}(\tau_n)
=
\sigma^2 V(P),
\label{eq:varentropy}
\end{equation}
where
\begin{equation}
V(P)
:=
\operatorname{Var}_P[-\ln P(h_n)]
=
\sum_{h_n} P(h_n)
\left[
-\ln P(h_n)-H(P)
\right]^2
\label{eq:varentropy-explicit}
\end{equation}
is the varentropy of the record distribution
\cite{Renyi1961,CoverThomas2006}. Thus the varentropy measures the dispersion of realized clock readings around
their ensemble mean: record histories with identical surprisal would have
$V(P)=0$, whereas heterogeneous history probabilities generate clock
fluctuations.

These results can be used to obtain the full probability law of the surprisal clock. For a fixed record depth $n$, the pathwise clock
\[
\tau_n(h)=-\sigma\ln P(h)
\]
is a random variable on the history space $\mathcal H_n$, and its
probability law can indeed be obtained from the R\'enyi entropy spectrum.

\begin{proposition}[Full law of the surprisal clock]
\label{prop:full-clock-law}
For a finite history ensemble, the probability measure of $\tau_n$ is
\begin{equation}
 \mu_{\tau_n}
 =
 \sum_{h\in\mathcal H_n}
 P(h)\,
 \delta_{-\sigma\ln P(h)} ,
 \label{eq:clock-pushforward}
\end{equation}
or equivalently
\begin{equation}
 \Pr(\tau_n\le t)
 =
 \sum_{\substack{h\in\mathcal H_n\\
 -\sigma\ln P(h)\le t}}P(h).
 \label{eq:clock-cdf}
\end{equation}
Its characteristic function is
\begin{equation}
 \Phi_n(\omega)
 :=
 \mathbb E[e^{i\omega\tau_n}]
 =
 \sum_{h\in\mathcal H_n}
 P(h)^{\,1-i\omega\sigma}.
 \label{eq:clock-characteristic}
\end{equation}
Hence $\Phi_n$ uniquely determines the full probability distribution of the
internal clock.  Wherever the R\'enyi entropy admits the corresponding
complex-order continuation, this may be written as
\begin{equation}
 \Phi_n(\omega)
 =
 \exp\!\left[
 i\omega\sigma\,
 H_{1-i\omega\sigma}(P)
 \right].
 \label{eq:clock-renyi-characteristic}
\end{equation}
\end{proposition}

\begin{proof}
Equation~\eqref{eq:clock-pushforward} follows directly by pushing the history
measure $P$ through the map $h\mapsto-\sigma\ln P(h)$.  Equation
\eqref{eq:clock-characteristic} then follows from
$e^{i\omega\tau_n(h)}=P(h)^{-i\omega\sigma}$.  A characteristic function
uniquely determines the probability law.  The R\'enyi representation follows
from
\[
 \sum_h P(h)^\alpha
 =
 \exp[(1-\alpha)H_\alpha(P)]
\]
with $\alpha=1-i\omega\sigma$, whenever this continuation is well defined.
\end{proof}

Thus the information spectrum of the record history becomes, after the
calibration by $\sigma$, a spectrum of possible internal durations.  Mean,
varentropy, higher cumulants, central-limit behavior, and large deviations
are different statistical properties or asymptotic regimes of
this underlying clock distribution. In the specific case of an i.i.d.\ record process with single-step law $p(r)$, the elementary
clock increment has the discrete law
\begin{equation}
 \mu_{\Delta\tau}
 =
 \sum_r p(r)\,
 \delta_{-\sigma\ln p(r)},
 \label{eq:single-step-clock-law}
\end{equation}
and the accumulated clock after $n$ records is the $n$-fold convolution
\begin{equation}
 \mu_{\tau_n}=\mu_{\Delta\tau}^{*n}.
 \label{eq:clock-convolution}
\end{equation}
This finite-$n$ law is exact; the Gaussian and large-deviation results provided below
describe, respectively, its typical and rare-fluctuation asymptotics.

\subsection{Actualization as information gain: the Doob decomposition of the clock}
\label{sec:information-gain}

The pathwise and ensemble descriptions of the clock can be unified by viewing
the accumulated surprisal as an adapted stochastic process on the record
filtration. Under Assumption~\ref{ass:persistence}, the record variables indeed define a classical
stochastic process governed by the Born measure $P$ on histories. Let
$\mathcal{F}_n$ denote the $\sigma$-algebra generated by
$(R_1,\ldots,R_n)$, and define the surprisal random variables
\begin{equation}
S_k:=-\ln p(R_k\mid R_1,\ldots,R_{k-1}),
\end{equation}
with realized values
\[
s_k=-\ln p(r_k\mid h_{k-1}),
\]
assuming $\mathbb{E}[S_k]<\infty$ for every $k$.

\begin{proposition}[Doob decomposition of the actualization clock]
\label{prop:doob-decomposition}
Define
\begin{align}
 \tau_n&:=\sigma\sum_{k=1}^{n}S_k,\notag\\
 A_n&:=\sigma\sum_{k=1}^{n}
 \mathbb{E}\!\left[S_k\mid\mathcal{F}_{k-1}\right]\notag\\
 &=\sigma\sum_{k=1}^{n}
 H\bigl(p(\cdot\mid R_1,\ldots,R_{k-1})\bigr),\notag\\
 M_n&:=\tau_n-A_n.
\label{eq:doob-objects}
\end{align}
Then $(A_n)$ is predictable (each increment is $\mathcal{F}_{n-1}$-measurable) and
nondecreasing with $A_0=0$, $(M_n)$ is an $(\mathcal{F}_n,P)$-martingale with $M_0=0$,
and
\begin{equation}
 \tau_n=A_n+M_n
\label{eq:doob-clock}
\end{equation}
is the unique such decomposition of the surprisal clock into a predictable nondecreasing
process started at zero and a martingale.
\end{proposition}

\begin{proof}
The increment $A_n-A_{n-1}=\sigma H\bigl(p(\cdot\mid R_1,\ldots,R_{n-1})\bigr)$ is a
function of $R_1,\ldots,R_{n-1}$ alone, hence $\mathcal{F}_{n-1}$-measurable, and it is
nonnegative because Shannon entropies are. Conditioning the clock increment on the past,
\begin{align}
 &\mathbb{E}\bigl[\tau_n-\tau_{n-1}\,\big|\,\mathcal{F}_{n-1}\bigr]\notag\\
 &\quad=\sigma\sum_{r}p(r\mid R_1,\ldots,R_{n-1})
 \bigl[-\ln p(r\mid R_1,\ldots,R_{n-1})\bigr]\notag\\
 &\quad=\sigma H\bigl(p(\cdot\mid R_1,\ldots,R_{n-1})\bigr)
 =A_n-A_{n-1}.
\end{align}
so that $\mathbb{E}[M_n-M_{n-1}\mid\mathcal{F}_{n-1}]=0$. Uniqueness is the standard
uniqueness of the Doob decomposition of an integrable adapted process.
\end{proof}

Taking expectations in Eq.~\eqref{eq:doob-clock} recovers
Eq.~\eqref{eq:conditional-entropy-sum}, since
$\mathbb{E}\bigl[H(p(\cdot\mid h_{k-1}))\bigr]=H(R_k\mid R_1,\ldots,R_{k-1})$. The
decomposition, however, contains additional information: the predictable process $A_n$ is the \emph{entropic compensator} of the
realized surprisal clock. Thus the pathwise clock decomposes into a
predictable, nondecreasing component, determined at each stage by the
conditional Shannon entropy of the next record, and a centered martingale
$M_n$ describing the fluctuations of the realized history around that
predictable component. Two consequences of this result are noteworthy. First, the
pathwise directionality of the clock is given a quantitative content.  Since $0<p\le1$, every increment satisfies $d(p)\ge0$, with equality exactly for a conditionally certain record, so accumulated record time cannot decrease along a history; the decomposition refines this orientation: the
predictable compensator $A_n$ carries the systematic growth, while the martingale part
$M_n$ carries the mean-zero fluctuations about it. Second, the clock-rate functional $\nu(\delta t)$ of
Section~\ref{sec:zeno-time}, Eq.~\eqref{eq:clock-rate}, is recognized as the compensator
rate per unit of laboratory time. The Zeno, Markovian, and anti-Zeno regimes classify
the growth of $A_n$, while the varentropy identity Eq.~\eqref{eq:varentropy} controls
the martingale part $M_n$.

The following corollary relates to the special regimes in which the clock
admits classical limit theorems. Note that neither
independence nor stationarity is assumed anywhere else in the paper.

\begin{corollary}[Law of large numbers and clock fluctuations]
\label{cor:clock-asymptotics}
If the record process is independent and identically distributed with per-step Born
distribution $p$, then $\tau_n/n\to\sigma H(p)$ almost surely, and
\begin{equation}
 \frac{\tau_n-n\sigma H(p)}{\sigma\sqrt{n\,V(p)}}
 \;\Longrightarrow\;\mathcal{N}(0,1),
\label{eq:clock-clt}
\end{equation}
where $V(p):=\operatorname{Var}\bigl[-\ln p(r)\bigr]\in(0,\infty)$ is the per-step varentropy, cf.\ Eq.~\eqref{eq:varentropy}. More generally,
for a stationary ergodic record process over a finite alphabet, with entropy rate
$h_\infty$, the Shannon--McMillan--Breiman theorem gives $\tau_n/n\to\sigma h_\infty$ almost surely and
in $L^1$: the actualization-time reading per record converges to $\sigma$ times the entropy rate of
the record process.
\end{corollary}

\begin{proof}
In the i.i.d.\ case the statements are the strong law of large numbers and the central
limit theorem applied to the i.i.d.\ random variables $S_k$, which have mean $H(p)$ and
variance $V(p)$. The stationary ergodic case is the Shannon--McMillan--Breiman theorem
\cite{CoverThomas2006} applied to $-n^{-1}\ln P(h_n)=\tau_n/(n\sigma)$.
\end{proof}

This asymptotic behavior provides a simple statistical mechanism by which a clock built from stochastic microscopic record increments can acquire an increasingly regular macroscopic reading, showing that microscopic stochasticity is compatible with macroscopic temporal regularity.

\section{Large deviations and the R\'enyi spectrum}
\label{sec:moments-deviations}

Proposition~\ref{prop:renyi} established that the R\'enyi entropy spectrum
determines the moment-generating function of the pathwise surprisal clock.
We now use the same structure to characterize rare fluctuations of the clock via the theory of large deviations.

For this, consider an i.i.d.\ record process with common single-step distribution $p$.
The surprisal increment at step $k$ is then
\begin{equation}
S_k=-\ln p(R_k),
\end{equation}
where the $R_k$ are i.i.d.\ record variables. Since the $S_k$ are themselves
i.i.d., let $S$ denote a generic single-step surprisal with the same
distribution. Its cumulant-generating function is
\begin{equation}
 \Lambda(\lambda)
 :=\ln\mathbb E[e^{\lambda S}]
 =\ln\sum_r p(r)^{1-\lambda}
 =\lambda H_{1-\lambda}(p),
 \label{eq:cgf}
\end{equation}
whenever the sum converges. Thus the full cumulant-generating function is
equivalently encoded by the R\'enyi entropy spectrum. Its derivatives at
$\lambda=0$ recover the ordinary cumulants discussed above, but we now focus instead on its dependence on $\lambda$ away from the origin.

For $N$ independent records, the accumulated dimensionless clock reading is
\begin{equation}
\frac{\tau_N}{\sigma}
=
\sum_{k=1}^{N}S_k,
\end{equation}
and therefore its cumulant-generating function is $N\Lambda(\lambda)$. The
average clock increment per record is
\begin{equation}
\bar S_N
:=
\frac{\tau_N}{N\sigma}
=
\frac{1}{N}\sum_{k=1}^{N}S_k.
\label{eq:empirical-clock-rate}
\end{equation}
By the law of large numbers, $\bar S_N$ approaches the typical value given by the Shannon entropy of the common single-step record distribution $H(p)$
as $N$ grows. Large-deviation theory addresses the complementary question:
how unlikely is a history for which $\bar S_N$ remains substantially different
from $H(p)$? For large $N$, such probabilities decay exponentially with $N$,
and the exponent is determined by the Legendre--Fenchel transform of
$\Lambda$. Since $\Lambda(\lambda)=\lambda H_{1-\lambda}(p)$, the R\'enyi
spectrum therefore determines not only the ordinary fluctuations of the
clock but also the asymptotic probabilities of atypically small or large
clock readings.

\begin{proposition}[Large deviations of the mean clock increment]
\label{prop:ldp}
Let the record steps be i.i.d.\ with common distribution $p$ on a finite
set of record outcomes. Then the mean dimensionless clock increment
\begin{equation}
\bar S_N
=
\frac{\tau_N}{N\sigma}
\end{equation}
satisfies Cram\'er's large-deviation principle. In particular, for large
$N$, the probability of observing a value $\bar S_N\simeq u$ away from its
typical value $H(p)$ decays exponentially as
\begin{equation}
\Pr(\bar S_N\simeq u)
\asymp
e^{-N I(u)},
\label{eq:ldp-scaling}
\end{equation}
where the convex rate function is
\begin{equation}
I(u)
=
\sup_{\lambda}
\left[
\lambda u-\Lambda(\lambda)
\right]
=
\sup_{\lambda}
\left[
\lambda u-\lambda H_{1-\lambda}(p)
\right].
\label{eq:rate-function}
\end{equation}
\end{proposition}

\begin{proof}
Because $p$ has finite support, the single-step surprisal
$S=-\ln p(R)$ takes only finitely many values, so its
cumulant-generating function $\Lambda(\lambda)$ is finite for all real
$\lambda$. Cram\'er's theorem therefore applies and gives the rate function
in Eq.~\eqref{eq:rate-function}. Since
$\Lambda'(0)=H(p)$ and $\Lambda''(0)=V(p)$, the rate function has its
minimum at $u=H(p)$ and, when $V(p)>0$, admits the local expansion
\[
I(u)
=
\frac{(u-H(p))^2}{2V(p)}
+
O\!\left((u-H(p))^3\right).
\]
If $V(p)=0$, the single-step surprisal is constant almost surely, so the
mean clock increment is deterministic.
\end{proof}

The large-deviation result gives a quantitative meaning to the stability of
the clock over long record histories. The value $H(p)$ is the typical
dimensionless clock increment per record, while fluctuations of order
$N^{-1/2}$ around this value describe the usual statistical variability of a
long clock history. The rate function $I(u)$ goes further by quantifying rare
histories in which the accumulated clock advances systematically more slowly
or more rapidly than its typical value. In particular, the analysis shows that their probabilities decrease
exponentially as $e^{-NI(u)}$. Thus the same record statistics that define
the surprisal clock also determine both its macroscopic regularity and the
probability of anomalous clock histories, without introducing any additional
clock postulate. 

\section{Record surprisal and quantum-state geometry}
\label{sec:projective-time}

The surprisal clock is defined from the conditional Born probabilities of
realized records. For projective records, these probabilities also determine
the geometric separation between the quantum states immediately before and
after the record is formed. This gives an exact relation between the
surprisal associated with a single record and the Fubini--Study geometry of
pure quantum states. This analysis distinguishes projective motion from the advance of the surprisal clock. 

A normalized pure state $|\psi\rangle$ represents a physical state only up to
an overall phase and therefore corresponds to a ray $[\psi]$ in projective
Hilbert space. The Fubini--Study metric provides the natural
unitary-invariant geometry on this space, and its geodesic distance between
two rays is determined by their quantum overlap
\cite{ProvostVallee1980,AnandanAharonov1990}. In what follows, we use the convention
\begin{equation}
d_{\rm FS}([\psi],[\phi])
=
\arccos|\langle\psi|\phi\rangle|,
\label{eq:fs-convention}
\end{equation}
in contrast to alternative conventions differing by an overall factor of two. Let $|\psi\rangle$ be the normalized state immediately before a projective
record is formed, and let $\Pi_r$ be the projector associated with the
realized record value $r$. Its Born probability is
\begin{equation}
p_r
=
\langle\psi|\Pi_r|\psi\rangle.
\end{equation}
For $p_r>0$, the corresponding conditional post-record state is
\begin{equation}
|\psi_r\rangle
=
\frac{\Pi_r|\psi\rangle}{\sqrt{p_r}}.
\label{eq:projective-update}
\end{equation}
We denote by
\begin{equation}
\theta_r
:=
d_{\rm FS}([\psi],[\psi_r])
\end{equation}
the Fubini--Study angle between the pre- and post-record rays.

\begin{proposition}[Surprisal and projective separation]
\label{prop:surprisal-fs}
For the projective record update in Eq.~\eqref{eq:projective-update},
\begin{equation}
\theta_r
=
\arccos\sqrt{p_r},
\qquad
-\ln p_r
=
-2\ln(\cos\theta_r).
\label{eq:surprisal-fs-exact}
\end{equation}
Consequently,
\begin{equation}
-\ln p_r
=
\theta_r^2
+
\frac{\theta_r^4}{6}
+
O(\theta_r^6),
\qquad
\theta_r\to0.
\label{eq:surprisal-fs-local}
\end{equation}
\end{proposition}

\begin{proof}
Since $\Pi_r^2=\Pi_r$,
\begin{equation}
\langle\psi|\psi_r\rangle
=
\frac{\langle\psi|\Pi_r|\psi\rangle}{\sqrt{p_r}}
=
\sqrt{p_r}.
\end{equation}
Equation~\eqref{eq:surprisal-fs-exact} then follows directly from the
definition of the Fubini--Study angle and from
$p_r=\cos^2\theta_r$. Expanding $-2\ln(\cos\theta_r)$ around
$\theta_r=0$ gives Eq.~\eqref{eq:surprisal-fs-local}.
\end{proof}

Thus, for projective records, the dimensionless surprisal increment is an
exact monotone function of the projective displacement produced by the
record update. In the near-certain regime $p_r\to1$,
\begin{equation}
-\ln p_r
=
d_{\rm FS}([\psi],[\psi_r])^2
+
O\!\left(d_{\rm FS}^4\right).
\label{eq:surprisal-fs-quadratic}
\end{equation}
The logarithmic and geometric quantities are therefore locally related, but
they describe different physical aspects of the process: the Fubini--Study
distance measures motion in projective Hilbert space, whereas the surprisal
clock advances only when a probabilistically nontrivial record is formed.

For comparison, along a differentiable unitary trajectory generated locally
by a Hamiltonian $H_T$, neighboring states satisfy
\begin{equation}
\left|
\langle\psi(T)|\psi(T+\dd T)\rangle
\right|^2
=
1-
\frac{(\Delta H_T)^2}{\hbar^2}\,\dd T^2
+
o(\dd T^2),
\end{equation}
while their Fubini--Study separation is \cite{AnandanAharonov1990}
\begin{equation}
\dd s_{\rm FS}
=
\frac{\Delta H_T}{\hbar}|\dd T|
+
o(|\dd T|).
\end{equation}
Hence the negative logarithm of the neighboring-state fidelity is second
order in $\dd T$, whereas projective distance is first order. This difference
becomes operational under repeated monitoring: in the Zeno limit, the
corresponding unitary segments can retain a finite accumulated Fubini--Study
length while the surprisal associated with the recorded survival outcomes
tends to zero, as shown in Sec.~\ref{sec:zeno-time}.

\section{Comparison and consistency of multiple clocks}
\label{sec:multiple}

When several record processes are used as clocks, their readings need not
coincide. This situation raises three distinct questions. First, how are
clock readings related when the underlying record filtration is changed,
either by recoding the same factual information or by combining additional,
possibly correlated, record streams? Second, for local quantum clocks, under
what conditions does additivity in a given measurement context extend to
additivity in every local context, and what does this imply about the
structure of the underlying quantum state? Third, when different clocks are
used to parametrize the same physical evolution, what consistency relations
can be established between their rates, conditional dynamics, and
quantum-state descriptions? We address these three questions in turn.

\subsection{Comparison of record filtrations}
\label{subsec:filtration-covariance}

The representation theorem applies to a specified record filtration. For a filtration
$\mathfrak F$, let $f$ denote a realized history of probability
$P_{\mathfrak F}(f)$ and define the dimensionless clock reading
\begin{equation}
s_{\mathfrak F}(f)
:=
\frac{\tau_{\mathfrak F}(f)}{\sigma_{\mathfrak F}}
=
-\ln P_{\mathfrak F}(f).
\label{eq:dimensionless-filtration-clock}
\end{equation}
The functional form of the clock is the same for all filtrations but its reading need not be, given that different filtrations may describe
either the same factual information or genuinely different record processes.

Suppose that two compatible filtrations $\mathfrak F$ and $\mathfrak G$
admit a common joint record filtration
$\mathfrak H=\mathfrak F\vee\mathfrak G$, meaning a record filtration whose
histories contain enough factual information to recover both the
$\mathfrak F$- and $\mathfrak G$-histories. For a realized joint history $(f,g)$, define
\begin{equation}
s_{\mathfrak H}(f,g)
:=
-\ln P_{\mathfrak H}(f,g),
\end{equation}
where $P_{\mathfrak H}(f,g)$ is the joint Born probability of obtaining the
$\mathfrak F$-history $f$ and the $\mathfrak G$-history $g$. For $(f,g)$ in the joint support, the probability chain rule gives
\begin{equation}
P_{\mathfrak H}(f,g)
=
P_{\mathfrak F}(f)P(g\mid f)
=
P_{\mathfrak G}(g)P(f\mid g),
\end{equation}
and therefore
\begin{equation}
s_{\mathfrak H}(f,g)
=
s_{\mathfrak F}(f)-\ln P(g\mid f)
=
s_{\mathfrak G}(g)-\ln P(f\mid g).
\label{eq:filtration-chain-comparison}
\end{equation}
Thus, whenever a joint record structure exists, the relation between the
dimensionless readings is fixed by the corresponding conditional Born
probabilities. Equivalently, define the pointwise mutual information (or information density)
between the two realized record histories by
\begin{equation}
\imath(f;g)
:=
\ln
\frac{P_{\mathfrak H}(f,g)}
{P_{\mathfrak F}(f)P_{\mathfrak G}(g)}.
\label{eq:filtration-pointwise-mutual-information}
\end{equation}
The joint clock reading then satisfies the general identity
\begin{equation}
s_{\mathfrak H}(f,g)
=
s_{\mathfrak F}(f)
+
s_{\mathfrak G}(g)
-
\imath(f;g).
\label{eq:general-filtration-additivity}
\end{equation}
Averaging over the joint history distribution gives
\begin{equation}
\mathbb E[s_{\mathfrak H}]
=
\mathbb E[s_{\mathfrak F}]
+
\mathbb E[s_{\mathfrak G}]
-
I(F:G),
\label{eq:mean-general-filtration-additivity}
\end{equation}
where
\begin{equation}
I(F:G)
=
\mathbb E[\imath(F;G)]
=
\sum_{f,g}
P_{\mathfrak H}(f,g)
\ln
\frac{P_{\mathfrak H}(f,g)}
{P_{\mathfrak F}(f)P_{\mathfrak G}(g)}
\geq 0
\label{eq:filtration-mutual-information}
\end{equation}
is the mutual information between the two record histories
\cite{CoverThomas2006}. Dimensional readings are recovered through
$\tau_{\mathfrak F}=\sigma_{\mathfrak F}s_{\mathfrak F}$ and
$\tau_{\mathfrak G}=\sigma_{\mathfrak G}s_{\mathfrak G}$, so distinct
physical calibrations remain possible.

One can now distinguish several important special cases. Suppose first that
$\mathfrak G$ is obtained from $\mathfrak F$ by an information-preserving
recoding $g=\phi(f)$. Then the probability distribution on $\mathfrak G$ is
the pushforward of that on $\mathfrak F$:
\begin{equation}
P_{\mathfrak G}(g)
=
\sum_{f:\,\phi(f)=g} P_{\mathfrak F}(f).
\end{equation}
For a bijective recoding, each $g=\phi(f)$ has a unique preimage, so that
\begin{equation}
P_{\mathfrak G}(\phi(f))
=
P_{\mathfrak F}(f).
\end{equation}
Consequently,
\begin{equation}
s_{\mathfrak G}(\phi(f))
=
s_{\mathfrak F}(f).
\label{eq:filtration-recoding-invariance}
\end{equation}
The dimensionless clock reading is therefore invariant under a bijective
redescription of the same record information. By contrast and more generally, suppose that $\mathfrak G$ contains the records of
$\mathfrak F$ together with an additional record stream $B$, so that
$g=(f,b)$. Then
\begin{equation}
s_{\mathfrak G}(f,b)
=
s_{\mathfrak F}(f)-\ln P(b\mid f),
\label{eq:strict-filtration-refinement}
\end{equation}
and therefore
\begin{equation}
\mathbb E[s_{\mathfrak G}]
-
\mathbb E[s_{\mathfrak F}]
=
H(B\mid F)
\geq0.
\label{eq:mean-filtration-refinement}
\end{equation}
The additional reading is precisely the conditional surprisal carried by
the factual distinction present in $\mathfrak G$ but absent from
$\mathfrak F$. In the specific case where the added stream is statistically independent of $\mathfrak F$, then
$\imath(f;b)=0$, and the general relation
Eq.~\eqref{eq:general-filtration-additivity} reduces to
\begin{equation}
s_{\mathfrak G}(f,b)
=
s_{\mathfrak F}(f)+s_B(b),
\qquad
s_B(b):=-\ln P_B(b),
\label{eq:independent-filtration-addition}
\end{equation}
and, for a common calibration $\sigma$,
\begin{equation}
\tau_{\mathfrak G}(f,b)
=
\tau_{\mathfrak F}(f)+\tau_B(b).
\label{eq:independent-filtration-time-addition}
\end{equation}
For statistically independent record streams, the joint clock is therefore additive: factorization of the joint Born weight translates, through the logarithmic clock law, into addition of the corresponding clock readings.

The preceding results are purely probabilistic: for a specified pair of
compatible record streams, clock additivity is equivalent to factorization of
their joint record distribution. Quantum mechanics allows a stronger question related to the analysis of the record distributions generated by a bipartite state under different local measurement contexts.

\subsection{Local clocks and context-universal additivity}

Consider a bipartite system whose two record streams arise
from local measurements on its two subsystems. In this setting the local
measurement contexts themselves can be varied, which leads to a stronger question: is clock additivity merely a property of the Born distribution in
one chosen pair of local contexts, or does it persist under every local change
of context?

To address this question, we first characterize additivity in a fixed pair of local contexts and then
determine when it holds universally. For this, let $A$
and $B$ be the corresponding local record variables, with joint distribution
$p_{ij}$ and marginals $p_i^A$ and $p_j^B$. Using a common calibration
$\sigma$, define
\begin{align}
d_{AB}(i,j)&=-\sigma\ln p_{ij},\\
d_A(i)&=-\sigma\ln p_i^A,\\
d_B(j)&=-\sigma\ln p_j^B,
\end{align}
for outcomes in the support of the corresponding distributions.

The general correlation identity
Eq.~\eqref{eq:general-filtration-additivity} then specializes to
\begin{equation}
d_{AB}(i,j)
=
d_A(i)+d_B(j)-\sigma\,\imath(i;j),
\label{eq:pointwise-clock-correlation}
\end{equation}
where
\begin{equation}
\imath(i;j)
=
\ln\frac{p_{ij}}{p_i^A p_j^B}
\label{eq:pointwise-mutual-information}
\end{equation}
is the pointwise mutual information. Hence, in a fixed pair of local
contexts, clock additivity holds for every outcome in the support if and only
if the corresponding Born distribution factorizes,
\begin{equation}
p_{ij}
=
p_i^A p_j^B.
\label{eq:fixed-context-factorization}
\end{equation}

Averaging Eq.~\eqref{eq:pointwise-clock-correlation} gives
\begin{equation}
D_{AB}
=
D_A+D_B-\sigma I(A:B),
\label{eq:mean-clock-correlation}
\end{equation}
and therefore
\begin{equation}
D_A+D_B-D_{AB}
=
\sigma I(A:B)
\geq0.
\label{eq:shared-clock-information}
\end{equation}
Thus the joint clock counts the shared information once, whereas adding the
two marginal readings counts it twice.

The fixed-context statement is purely probabilistic. In particular,
factorization of the outcome distribution for one chosen pair of local
measurement contexts does not imply that the underlying quantum state is a
product state. A stronger statement becomes possible only when factorization
is required to hold under all local changes of measurement context. For a finite-dimensional bipartite state $\rho_{AB}$, a standard
characterization of product states is that
\begin{equation}
\rho_{AB}
=
\rho_A\otimes\rho_B
\end{equation}
if and only if the joint Born probabilities factorize for every pair of local
measurement contexts. Equivalently, it is sufficient to require
factorization for all pairs of local projective contexts, since the
corresponding local measurement statistics are tomographically complete
\cite{Horodecki2009,Watrous2018}. Combining this standard characterization with
Eq.~\eqref{eq:pointwise-clock-correlation} gives an immediate consequence with respect to clock additivity.

\begin{corollary}[Context-universal clock additivity]
\label{cor:context-universal}
For a finite-dimensional bipartite quantum state, with a common clock
calibration $\sigma$, the following statements are equivalent:
\begin{enumerate}
\item[(i)] The state is a product state,
\begin{equation}
\rho_{AB}
=
\rho_A\otimes\rho_B.
\end{equation}

\item[(ii)] For every pair of local measurement contexts, the joint
surprisal clock is additive for every outcome in its support,
\begin{equation}
d_{AB}^{U,V}(i,j)
=
d_A^{U,V}(i)+d_B^{U,V}(j).
\label{eq:universal-clock-additivity}
\end{equation}
\end{enumerate}
\end{corollary}

\begin{proof}
Equation~\eqref{eq:pointwise-clock-correlation} shows that, in any fixed pair
of local contexts, clock additivity is equivalent to factorization of the
corresponding Born distribution. By the standard product-state characterization above, requiring this factorization for every pair
of local contexts is
equivalent to $\rho_{AB}=\rho_A\otimes\rho_B$.
\end{proof}

Additivity in one pair of local contexts establishes only
probabilistic independence of the corresponding observed outcomes.
Context-universal additivity is much stronger: it holds if and only if the
underlying bipartite state itself contains no correlations, that is, if and
only if it is a product state. For a non-product state, there necessarily exist local
contexts in which the joint surprisal clock reading contains a nonzero correlation
correction.

\subsection{Comparing different clock descriptions}
\label{subsec:clock-comparison}

The preceding results concern the information carried by different record
streams. When two record clocks parametrize the same underlying physical evolution, the natural question is how to compare their respective rates.

Let two ensemble record clocks be indexed by a common parameter $\lambda$,
with accumulated-history entropies $H_A(\lambda)$ and $H_B(\lambda)$. By
Eq.~\eqref{eq:mean-clock}, their mean readings are
\begin{equation}
\bar\tau_A(\lambda)
=
\sigma_A H_A(\lambda),
\qquad
\bar\tau_B(\lambda)
=
\sigma_B H_B(\lambda).
\label{eq:two-ensemble-clocks}
\end{equation}
Whenever these functions are differentiable and
$d\bar\tau_A/d\lambda\neq0$, their relative rate is
\begin{equation}
\frac{d\bar\tau_B}{d\bar\tau_A}
=
\frac{\sigma_B}{\sigma_A}
\frac{dH_B/d\lambda}{dH_A/d\lambda}.
\label{eq:relative-rates}
\end{equation}
Although the individual rates with respect to $\lambda$ depend on its
parametrization, their ratio does not: under any smooth monotone
reparametrization $\mu=\mu(\lambda)$, the common Jacobian factor cancels.
Equation~\eqref{eq:relative-rates} therefore provides a
reparametrization-invariant comparison between the mean readings of two
record clocks.

The present construction determines such clock readings and their relative
rates, but does not by itself derive the conditional dynamics of other
degrees of freedom with respect to them. Whether a suitable continuous
record clock can serve as the parameter of relational Schr\"odinger evolution
requires additional dynamical assumptions and is left for further study.

\section{Measurement frequency and the rate of the surprisal clock}
\label{sec:zeno-time}

The surprisal clock advances only when a record resolves genuine uncertainty:
an outcome of conditional probability one contributes no clock increment.
The influence of repeated monitoring must therefore be assessed not only
through its effect on quantum dynamics, but also through the record
probabilities from which the clock is constructed. We analyze this dependence
in the Zeno regime of quantum measurement theory
\cite{MisraSudarshan1977}, in a Markovian monitoring regime, and in an
anti-Zeno regime.

\subsection{The clock-rate functional}

Let $\delta t$ be the laboratory interval between successive records. This
quantity enters only as an external control parameter and plays no role in the
definition of the internal-time increment itself. Let
\begin{equation}
\eta_n
=
H(R_n\mid R_1,\ldots,R_{n-1})
\end{equation}
be the conditional entropy of the next record. By
Proposition~\ref{prop:shannon-mean}, the mean clock advance per round is
$\sigma\eta_n$. We therefore define the mean clock rate relative to the
external laboratory parameter by
\begin{equation}
\nu(\delta t)
:=
\frac{\sigma\,\eta_n}{\delta t}.
\label{eq:clock-rate}
\end{equation}

For the Zeno analysis, we specialize to binary survival monitoring. Each
interrogation produces one of two records: survival, with probability
$1-q(\delta t)$, or non-survival, with probability $q(\delta t)$. The
corresponding realized clock increment is
\begin{equation}
\Delta\tau
=
\begin{cases}
-\sigma\ln[1-q(\delta t)],
& \text{survival},\\[1mm]
-\sigma\ln q(\delta t),
& \text{non-survival}.
\end{cases}
\end{equation}
Its ensemble mean per interrogation is therefore
\begin{equation}
\mathbb E[\Delta\tau]
=
\sigma H_{\rm bin}(q(\delta t)),
\end{equation}
where
\begin{equation}
H_{\rm bin}(q)
:=
-q\ln q-(1-q)\ln(1-q)
\label{eq:binary-entropy}
\end{equation}
is the binary Shannon entropy. Relative to the externally controlled
laboratory interval $\delta t$, the mean internal-time accumulation rate is
thus
\begin{equation}
\nu(\delta t)
=
\frac{\sigma H_{\rm bin}(q(\delta t))}{\delta t}.
\label{eq:clock-rate-binary}
\end{equation}

It is useful to introduce the effective non-survival rate
\begin{equation}
\gamma_{\rm eff}(\delta t)
:=
\frac{q(\delta t)}{\delta t}.
\end{equation}
In the Zeno regime, quadratic short-time survival gives
$q(\delta t)=O(\delta t^2)$ and hence
$\gamma_{\rm eff}\to0$. In a Markovian regime,
$q(\delta t)\sim\gamma\delta t$ and
$\gamma_{\rm eff}\to\gamma>0$. An anti-Zeno window corresponds instead
to a monitoring-induced enhancement of the effective transition rate.

\subsection{Zeno collapse of the surprisal clock}
\label{subsec:zeno}

Consider rank-one survival monitoring. Starting from $|\psi\rangle$, the
system evolves under a self-adjoint Hamiltonian $\widehat H$ for a time
$\delta t$ and is then measured in the binary context
$\{\Pi_\psi,\mathbb I-\Pi_\psi\}$, with
$\Pi_\psi=|\psi\rangle\langle\psi|$. The short-time survival law follows
directly from unitary dynamics. Writing
\begin{equation}
A(\delta t)
=
\langle\psi|
e^{-i\widehat H\delta t/\hbar}
|\psi\rangle,
\end{equation}
expansion around $\delta t=0$ gives
\begin{equation}
|A(\delta t)|^2
=
1-
\frac{(\Delta E)^2}{\hbar^2}\delta t^2
+
O(\delta t^4),
\end{equation}
where the energy uncertainty is
\begin{equation}
(\Delta E)^2
:=
\langle\widehat H^2\rangle
-
\langle\widehat H\rangle^2.
\end{equation}
Hence the non-survival probability satisfies
\begin{equation}
q(\delta t)
=
\frac{(\Delta E)^2}{\hbar^2}\,\delta t^2
+
O(\delta t^4)
=
\left(\frac{\delta t}{\tau_Z}\right)^2
+
O(\delta t^4),
\label{eq:short-time-survival}
\end{equation}
where
\begin{equation}
\tau_Z
=
\frac{\hbar}{\Delta E}
\end{equation}
is the standard Zeno time. For later convenience, define
\begin{equation}
\alpha_Z
:=
\frac{(\Delta E)^2}{\hbar^2}
=
\frac{1}{\tau_Z^2}.
\label{eq:zeno-alpha}
\end{equation}

On the survival branch, the elementary internal-time increment therefore
obeys
\begin{equation}
\delta\tau_{\rm surv}
=
-\sigma\ln[1-q(\delta t)]
=
\sigma\alpha_Z\,\delta t^2
+
o(\delta t^2)
=
\sigma
\left(
\frac{\Delta E\,\delta t}{\hbar}
\right)^2
+
o(\delta t^2).
\label{eq:zeno-elementary-clock}
\end{equation}
Thus the energy uncertainty fixes the microscopic scale on which unitary
evolution generates nontrivial record surprisal.

For a sequence of $N$ recorded interrogations, let
\begin{equation}
\tau_N
=
\sum_{n=1}^{N}\delta\tau_n
\end{equation}
denote the random accumulated internal clock. For a particular realized
record history $h$, we write $\tau(h)$ for the corresponding realization
of $\tau_N$. In particular,
\begin{equation}
h_{\rm surv}
=
(\mathrm{surv},\ldots,\mathrm{surv})
\end{equation}
denotes the history in which all $N$ interrogations return the survival
outcome.

\begin{proposition}[Zeno collapse of the surprisal clock]
\label{prop:zeno-freeze}
Fix a laboratory interval $T>0$ and let
\begin{equation}
N
=
\left\lfloor\frac{T}{\delta t}\right\rfloor.
\end{equation}
As $\delta t\to0$,
\begin{align}
P(h_{\rm surv})
&=
(1-q)^N
\longrightarrow
1,
\\
\tau(h_{\rm surv})
&=
-\sigma N\ln(1-q)
=
\sigma\alpha_Z T\,\delta t\,(1+o(1))
\longrightarrow
0,
\label{eq:zeno-pathwise}
\\
\mathbb E[\tau_N]
&=
2\sigma\alpha_Z T\,\delta t
\ln\frac{1}{\delta t}\,(1+o(1))
\longrightarrow
0.
\label{eq:zeno-mean}
\end{align}
Moreover,
\begin{equation}
\tau_N
\xrightarrow{\;P\;}
0
\qquad\text{while}\qquad
N\longrightarrow\infty.
\label{eq:zeno-clock-collapse}
\end{equation}
\end{proposition}

\begin{proof}
From Eq.~\eqref{eq:short-time-survival},
\begin{equation}
q(\delta t)
=
\alpha_Z\,\delta t^2+O(\delta t^4),
\end{equation}
while
\begin{equation}
N
=
\frac{T}{\delta t}+O(1).
\end{equation}
Hence
\begin{equation}
Nq
=
\alpha_Z T\,\delta t\,(1+o(1)).
\end{equation}
Since $\ln(1-q)=-q+O(q^2)$, this immediately gives
Eq.~\eqref{eq:zeno-pathwise} and
\begin{equation}
P(h_{\rm surv})
=
\exp\!\left[N\ln(1-q)\right]
\longrightarrow1.
\end{equation}
For binary survival monitoring, the mean surprisal per interrogation is
$\sigma H_{\rm bin}(q)$. As $q\to0$,
\begin{equation}
H_{\rm bin}(q)
=
q\left[1+\ln(1/q)\right]+O(q^2).
\label{eq:binary-entropy-small-q}
\end{equation}
Using
$q=\alpha_Z\delta t^2(1+o(1))$ and
$N=T/\delta t+O(1)$ gives
\begin{equation}
\mathbb E[\tau_N]
=
N\sigma H_{\rm bin}(q)
=
2\sigma\alpha_Z T\,\delta t
\ln\frac{1}{\delta t}\,(1+o(1)),
\end{equation}
which proves Eq.~\eqref{eq:zeno-mean}. Finally, the probability of at least one non-survival record is
\begin{equation}
1-(1-q)^N
=
\alpha_Z T\,\delta t\,(1+o(1))
\longrightarrow0.
\end{equation}
On the complementary all-survival history,
$\tau(h_{\rm surv})\to0$. Therefore, for every $\varepsilon>0$,
\begin{equation}
P(\tau_N>\varepsilon)
\leq
P(\text{at least one non-survival record})
+
o(1)
\longrightarrow0,
\end{equation}
which establishes convergence in probability.
\end{proof}

Note that the collapse of the surprisal clock does not imply that the underlying
unitary motion disappears. During each free segment, the Fubini--Study
angle between the initial and evolved rays is
\begin{equation}
\theta(\delta t)
=
\arccos\sqrt{1-q(\delta t)}
=
\sqrt{\alpha_Z}\,\delta t
+
O(\delta t^3),
\end{equation}
and therefore
\begin{equation}
\sum_{n=1}^{N}\theta(\delta t)
\longrightarrow
\sqrt{\alpha_Z}\,T
=
\frac{\Delta E}{\hbar}\,T.
\label{eq:zeno-projective-length}
\end{equation}

An arbitrarily fine sequence of recorded interrogations can therefore
generate a diverging ordinal record count while accumulating vanishing
internal duration. The reason is that the repeated survival records become
asymptotically certain and hence carry asymptotically vanishing informational
novelty. In this precise sense, the Zeno limit separates three notions that
need not coincide: the ordinal accumulation of records, dynamical motion in
projective state space, and the metric duration generated by record
surprisal.

\subsection{Continuous monitoring and anti-Zeno acceleration}
\label{subsec:rate-antizeno}

The quadratic short-time law used above describes the microscopic unitary
regime underlying the Zeno effect. A different scaling arises in an effective
Markovian regime, as commonly obtained for an open quantum system after
coarse-graining over environmental degrees of freedom. If the reduced
dynamics has a finite transition rate $\gamma$, the non-survival probability
over a short monitoring interval within this regime has the form
\begin{equation}
q(\delta t)
=
\gamma\,\delta t+O(\delta t^2),
\qquad
\gamma>0.
\label{eq:markovian-transition}
\end{equation}
Equivalently, this is the local form of exponential survival,
$P_{\rm surv}(t)\sim e^{-\gamma t}$. In this case, using
Eq.~\eqref{eq:binary-entropy-small-q} in
Eq.~\eqref{eq:clock-rate-binary} gives
\begin{equation}
\nu(\delta t)
=
\sigma\gamma
\left[
\ln\frac{1}{\gamma\delta t}+1
\right]
+
O(\delta t).
\label{eq:rate-regime}
\end{equation}
The surprisal-clock rate therefore depends logarithmically on the monitoring
resolution. At a fixed physical resolution $\delta t_\ast$ it remains finite;
the divergence as $\delta t\to0$ reflects the accumulation of increasingly
certain null records when arbitrarily fine monitoring is itself included in
the clock filtration.

More precisely, let $X$ denote the monitored system variable. The conditional
record entropy can be decomposed by the chain rule as
\begin{equation}
H(R_n\mid R_{<n})
=
I(X:R_n\mid R_{<n})
+
H(R_n\mid X,R_{<n}),
\label{eq:monitoring-entropy-decomposition}
\end{equation}
where the first term measures information newly acquired about the system and
the second is conditional apparatus noise.

\begin{proposition}[Bound on system-attributable monitoring time]
\label{prop:noise-funded}
For a discrete monitored variable $X$, define the system-attributable part of
the clock by
\begin{equation}
\tau_{\rm sys}(N)
:=
\sigma\sum_{n=1}^{N}I(X:R_n\mid R_{<n}).
\end{equation}
Then
\begin{equation}
\tau_{\rm sys}(N)
=
\sigma\,I(X:R_1,\ldots,R_N)
\leq
\sigma H(X).
\label{eq:system-clock-bound}
\end{equation}
\end{proposition}

\begin{proof}
The identity follows from the chain rule for mutual information, and the
inequality from
$I(X:R_1,\ldots,R_N)\leq H(X)$.
\end{proof}

Arbitrarily fine sampling can therefore generate unboundedly many records
whose total surprisal grows without limit, while the component attributable
to learning a fixed system variable remains bounded. In a non-Markovian
regime, the effective rate
$\gamma_{\rm eff}(\delta t)=q(\delta t)/\delta t$
can increase as $\delta t$ decreases, for example through the interaction of
measurement-induced spectral broadening with a structured reservoir
\cite{KofmanKurizki2000}. In such an anti-Zeno window,
\begin{equation}
\nu(\delta t)
\simeq
\sigma\gamma_{\rm eff}(\delta t)
\left[
\ln\frac{1}{q(\delta t)}+1
\right],
\end{equation}
so more frequent interrogation can accelerate the surprisal clock. At still
shorter intervals, however, the universal quadratic survival law is recovered
and the Zeno collapse resumes.

The Zeno, effective Markovian, and anti-Zeno regimes therefore exhibit distinct behaviors of the mean surprisal produced per unit laboratory time: it vanishes in the Zeno limit, displays a logarithmic dependence on monitoring resolution within the Markovian regime, and can be enhanced over an anti-Zeno situation.

\section{Discussion and outlook}
\label{sec:further-research}
The preceding sections established the formal properties of the
record-defined surprisal clock and examined several of its statistical and
physical consequences. We now discuss the ontological interpretation suggested by the growing
record structure, and the main open problems that remain for a predictive
theory of physical clocks.

\subsection{Ontological interpretation}

The formal results of the paper do not require an ontological interpretation
of actualization but they nevertheless suggest one naturally. By
Theorem~\ref{thm:growing-records}, persistent records form a nested sequence
of valued algebras,
\begin{equation}
\Sigma_n\subseteq\Sigma_{n+1},
\qquad
v_{n+1}|_{\Sigma_n}=v_n,
\end{equation}
so that determinate factual content can increase while previously recorded
facts remain fixed. This structure is naturally related to growing-block
views of time, in which factual reality is taken to accumulate rather than
being either restricted to the present or fixed as a complete past--present--future
block \cite{Broad1923,EllisDrossel2020}. The potentiality interpretation developed in
Ref.~\cite{MartelliniPotentiality} adds an important qualification to this
analogy. What has not yet actualized need not be regarded as nonexistent:
unresolved alternatives may remain physically real as potentialities without
yet belonging to the growing domain of factual actuality. The resulting
picture may therefore be described as a \emph{growing block of actuality
relative to an underlying atemporal block of potentialities}. 

This picture may appear to be inconsistent with
the usual Schr\"odinger description, in which an unresolved quantum state can evolve
continuously with respect to an externally supplied parameter,
\[
i\hbar\frac{\partial}{\partial t}|\psi(t)\rangle
=
H|\psi(t)\rangle,
\]
even when no new outcome is actualized and no persistent record is formed.
This does not constitute a problem with Schr\"odinger dynamics as a physical
description. Rather, it shows that an externally parametrized Schr\"odinger
evolution cannot by itself explain the origin of time, since the temporal
parameter has already been supplied before either the evolution of
potentialities or their actualization is described. The potentiality--actuality
distinction and the temporal past--future distinction must therefore be kept
logically separate.

Potentiality should therefore not be understood as a static future awaiting
conversion into an actualized past against the background of an independently
given time parameter. The interpretation proposed here reverses this explanatory order. The
underlying domain of unresolved potentialities is not assigned a prior
temporal location or assumed to evolve fundamentally with respect to an
external clock. Actualization gives rise to persistent factual distinctions;
the inclusion structure of those distinctions provides an ordinal ordering;
and the statistics of their successive formation determine the quantitative
clock constructed in this paper. Temporal structure is therefore reconstructed
from the organization and accumulation of actualized records rather than
assumed as a background in which potentialities first evolve and subsequently
actualize. Within the standard Hilbert-space formalism of the present paper, the relevant
operational content is represented by the sequence of conditional states
\begin{equation}
|\psi\rangle
\longrightarrow
|\psi_{h_1}\rangle
\longrightarrow
|\psi_{h_2}\rangle
\longrightarrow\cdots .
\label{eq:ontological-conditioned-state-sequence}
\end{equation}
On the potentiality reading, this sequence can be interpreted as the successive
conditioning of unresolved potentialities by the records that have become
factual. The claim is thus not that a
fundamental state must evolve with respect to an already existing external time
coordinate. Rather, each $|\psi_{h_n}\rangle$ is the predictive quantum state
relative to a progressively richer factual history, and the order of that
history has itself been supplied by record inclusion. 

The atemporality attributed to the underlying domain of potentialities is
compatible with a history-dependent predictive state. What changes
operationally is the conditioning of the quantum description on newly
actualized facts; what grows ontologically, in this interpretation, is the
domain of factual actuality. Their common ordering is the record order derived
above rather than an independently postulated temporal parameter. In this sense, the update is Bayesian-like, but it is not classical Bayesian
conditioning on a pre-existing assignment of values since the conditioned quantum
state retains the coherent and contextual structure of the unresolved
alternatives.

At the interpretive level, this viewpoint is closer in spirit to timeless
formulations of generally covariant quantum theory, in which the fundamental
description need not contain a preferred external time parameter. This analogy
is limited, however, since no Wheeler--DeWitt equation or particular
quantum-gravitational dynamics is assumed here. The formal results established
in this paper require only the persistent record structure and conditional
Born probabilities described above, and therefore do not depend on this
ontological or any specific quantum-gravitational interpretation.

\subsection{Open questions}

A key open problem is the microscopic construction of physical clocks.  The
representation theorem determines the clock reading once a constitutive
record process is specified, but a physical model must still explain which
degrees of freedom generate persistent records, what fixes their effective
resolution, and what dynamics controls their production rate.  A complete
material-clock model should derive these ingredients from a concrete quantum
system and determine the conditions under which the resulting record process
is sufficiently stable and regular to support a robust macroscopic clock.
We leave this dynamical construction for further research.

A second question is whether familiar macroscopic time laws can emerge from
such record dynamics rather than being imposed as external rate laws.  In
particular, one would like to determine when the accumulation of microscopic
record surprisal approaches a regular macroscopic clock and how its residual
fluctuations scale under coarse graining.  The entropy, varentropy,
R\'enyi-spectrum, and martingale structure obtained above provide natural
starting points for such a stochastic macroscopic limit. A central test is then whether relativistic clock-rate relations can emerge
from the same microscopic mechanism.  In special relativity, the relevant
question is whether the proper-time rate of a moving material clock can be
derived from changes in its record-generation statistics, so that the
standard velocity-dependent time-dilation law emerges rather than being
inserted externally.  In a gravitational setting, one should similarly ask
whether gravitational clock redshift and, ultimately, the local proper-time
structure of general relativity can arise from the coupling between record
dynamics and an emergent relational geometry.  This requires a joint theory
of clock dynamics, motion, and spatial structure, and is left for further
research.

A related problem concerns transformations between distinct clocks.
Section~\ref{subsec:filtration-covariance} gives exact comparison maps for
compatible record filtrations with a common joint refinement, while
Section~\ref{subsec:clock-comparison} shows that ensemble clock rates can be
compared through a reparametrization-invariant relative rate.  A fuller theory should characterize
admissible clock descriptions and filtrations, comparison maps on overlaps,
consistency on multiple overlaps, and possible obstructions to a single
global temporal description.  This becomes especially relevant for noisy,
entangled, nonideal, or operationally incompatible clocks, for which locally
meaningful clock readings need not assemble into a unique global time. More broadly, the distinction established here between ordinal succession
and metric duration suggests asking whether temporal and spatial relations can
be reconstructed from different aspects of the same record structure.
Temporal duration is associated with the informational novelty of a newly
actualized record relative to the accumulated record history, whereas
relational separation should depend on the information shared, or not shared,
between distinct records conditioned on a common prior record history.  Whether these constructions
can jointly generate an effective Lorentzian spacetime geometry is an open
problem.

A complementary operational direction concerns time-of-arrival and related
temporal observables.  If time is assigned through actualized records rather
than introduced only as an external parameter, one should determine whether
arrival-time statistics can be reconstructed directly from detector-record
processes, whether the framework selects a preferred class of admissible
time-of-arrival distributions, and how these predictions relate to standard
quantum constructions.

\section{Conclusion}
\label{sec:conclusion}

This paper has developed a record-based notion of internal time in which
persistent records first provide an ordinal chronology through inclusion of
their accumulated record algebras. On a specified record filtration,
Born-weight sufficiency, record-refinement invariance and regularity uniquely
fix, up to an overall positive calibration scale, the pathwise clock reading as
\begin{equation}
\tau_{\rm A}(h)
=
-\sigma\ln P(h).
\end{equation}
Shannon entropy is its ensemble mean, and the varentropy, R\'enyi spectrum,
Doob decomposition, and large-deviation structure characterize the whole distribution beyond the mean.

The construction is universal in functional form but relative to a physically
specified record filtration. It does not determine which possible outcome
becomes factual, nor does it select which physical record process constitutes
a material clock or determine the rate at which its records are produced.
Those are dynamical questions. Once a constitutive record process is specified,
however, the theorem fixes how its realized Born probabilities are converted
into an additive internal-time reading.

Several results test the consistency of this construction beyond the
representation theorem itself. For finite-dimensional bipartite systems,
universal additivity of local clock readings across admissible local
measurement contexts characterizes product structure within the stated class
of states. For projective records, surprisal is an exact monotone function of
the Fubini--Study separation generated by the record update,
\begin{equation}
-\ln p=-2\ln\cos\theta,
\end{equation}
while the Zeno limit demonstrates that projective motion, ordinal record
accumulation, and surprisal duration are distinct: the number of recorded
interrogations can diverge and the accumulated Fubini--Study length of the
intervening unitary segments remain finite while the surprisal clock converges
to zero.

The main result is therefore not only a representation theorem and also a separation
of roles: persistent records supply temporal order, their conditional Born
weights uniquely determine the refinement-invariant additive clock reading
within the stated class, and microscopic dynamics must supply the physical
process that generates those records and sets their rate. Specific quantitative
predictions for material clocks remain to be derived from concrete physical
realizations.

\ifarxiv
  \section*{Statements and Declarations}

  \subsection*{Funding}
  The author received no specific funding for this work.

  \subsection*{Competing interests}
  The author declares no competing interests.

  \subsection*{Use of artificial intelligence tools}
  Large language model assistants (OpenAI ChatGPT and Anthropic Claude) were
  used during the preparation of this work, under the author's direction, for
  drafting and revising portions of the text, checking derivations and internal
  consistency, suggesting and verifying references, and preparing the LaTeX
  source. All definitions, results, proofs, and interpretive claims were
  critically assessed and approved by the author, who takes full responsibility
  for the content of the manuscript.

  \subsection*{Data availability}
  Data sharing is not applicable to this article as no datasets were generated or analyzed during the current study.

  \bibliographystyle{apsrev4-2}
\else
  \backmatter

  \section*{Statements and Declarations}

  \bmhead{Funding}
  The author received no specific funding for this work.

  \bmhead{Competing interests}
  The author declares no competing interests.

  \bmhead{Use of artificial intelligence tools}
  Large language model assistants (OpenAI ChatGPT and Anthropic Claude) were
  used during the preparation of this work, under the author's direction, for
  drafting and revising portions of the text, checking derivations and internal
  consistency, suggesting and verifying references, and preparing the LaTeX
  source. All definitions, results, proofs, and interpretive claims were
  critically assessed and approved by the author, who takes full responsibility
  for the content of the manuscript.

  \bmhead{Data availability}
  Data sharing is not applicable to this article as no datasets were generated or analyzed during the current study.
\fi

\bibliography{2.1referencesTime}

\end{document}